\documentclass[letterpaper,12pt,oneside,final]{article}

\usepackage{amsmath, amssymb,mathrsfs,amsthm, complexity, mathtools}
\usepackage{tikz}
\usepackage{tkz-graph}
\usetikzlibrary{shapes,backgrounds}
\usepackage[ruled,vlined]{algorithm2e}
\newtheorem{theorem}{Theorem}[section]
\newtheorem{lemma}[theorem]{Lemma}

\newtheorem{definition}{Definition}
\newtheorem{problem}{Problem}[]
\newtheorem{observation}{Observation}

\newtheoremstyle{case}{}{}{}{}{}{:}{ }{}
\theoremstyle{case}
\newtheorem{case}{Case}
\newtheorem{case1}{Case}
\newtheorem{case2}{Case}
\newtheorem{case3}{Case}
\newtheorem{case4}{Case}
\newtheorem{case5}{Case}
\newtheorem{case6}{Step}

\title{Building a larger class of graphs for efficient reconfiguration of vertex colouring}
\author{Therese Biedl
\thanks{Work of TB and AL supported by NSERC.} 
\and Anna Lubiw \addtocounter{footnote}{-1}\footnotemark
\and Owen Merkel
}
\date{}
\begin{document}
\maketitle

\begin{center}
{\footnotesize
David R. Cheriton School of Computer Science, University of Waterloo, \\Waterloo, Ontario, Canada 
\\{\tt \{biedl, alubiw, odmerkel\}@uwaterloo.ca}}
\end{center}

\begin{abstract}

A $k$-colouring of a graph $G$ is an assignment of at most $k$ colours to the vertices of $G$ so that adjacent vertices are assigned different colours. The reconfiguration graph of the $k$-colourings, $\mathcal{R}_k(G)$, is the graph whose vertices are the $k$-colourings of $G$ and two colourings are joined by an edge in $\mathcal{R}_k(G)$ if they differ in colour on exactly one vertex. For a $k$-colourable graph $G$, we investigate the connectivity and diameter of $\mathcal{R}_{k+1}(G)$. It is known that not all weakly chordal graphs have the property that $\mathcal{R}_{k+1}(G)$ is connected. On the other hand, $\mathcal{R}_{k+1}(G)$ is connected and of diameter $O(n^2)$ for several subclasses of weakly chordal graphs such as chordal, chordal bipartite, and $P_4$-free graphs.

We introduce a new class of graphs called OAT graphs that extends the latter classes and in fact extends outside the class of weakly chordal graphs. OAT graphs are built from four simple operations, disjoint union, join, and the addition of a clique or comparable vertex. We prove that if $G$ is a $k$-colourable OAT graph then $\mathcal{R}_{k+1}(G)$ is connected with diameter $O(n^2)$. Furthermore, we give polynomial time algorithms to recognize OAT graphs and to find a path between any two colourings in $\mathcal{R}_{k+1}(G)$.
\end{abstract}

\section{Introduction}
A reconfiguration framework consists of states and transitions between states. The states represent feasible solutions to a source problem and there is a transition between states if the corresponding feasible solutions satisfy a predefined adjacency relationship. Reconfiguration problems have been studied in many areas (see e.g. \cite{nishimura2018}) and have a wide-range of applications. Here we focus on reconfiguration of vertex colourings.

The reconfiguration graph of the $k$-colourings, $\mathcal{R}_k(G)$, is the graph whose vertices are the $k$-colourings of $G$ and two colourings are joined by an edge in $\mathcal{R}_k(G)$ if they differ in colour on exactly one vertex. Several problems have been studied in the area of reconfiguration for vertex colouring, two of which are concerned with the connectivity and diameter of $\mathcal{R}_k(G)$. The first problem is to determine if for any two colourings, it is always possible to reconfigure one colouring into the other by recolouring a single vertex at a time and while remaining a proper colouring at each step. The second problem is to determine an upper bound on the number of recolourings needed to reconfigure one colouring to another.

Bonamy et al.~\cite{bonamy2014} asked the following question. Given a $k$-colourable perfect graph $G$, is $\mathcal{R}_{k+1}(G)$ connected with diameter $O(n^2)$? One cannot hope for a smaller diameter since for $P_n$, the path on $n$ vertices, $\mathcal{R}_3(P_n)$ has diameter $\Omega(n^2)$ \cite{bonamy2014}. Bonamy and Bousquet \cite{bonamy2018} answered this question in the negative by showing $\mathcal{R}_{k+1}(G)$ has an isolated vertex when $G$ is a complete bipartite graph minus a matching and where $k$ can be arbitrarily large. Feghali and Fiala \cite{feghali2020} also investigated this question and found an infinite family of weakly chordal graphs $G$ where $\mathcal{R}_{k+1}(G)$ has an isolated vertex. In the same paper, Feghali and Fiala introduced a subclass of weakly chordal graphs called $\emph{compact}$ graphs (definitions and details in Section \ref{sec:background}). They prove that for a $k$-colourable compact graph $G$, $\mathcal{R}_{k+1}(G)$ is connected with diameter $O(n^2)$.

In this paper, we consider a class of graphs which we call \emph{OAT} graphs that can be constructed from four simple operations, defined as follows. Let $G_1=(V_1, E_1)$ and $G_2=(V_2, E_2)$ be vertex-disjoint OAT graphs.

\begin{definition}
A graph $G$ is an OAT graph if it can be constructed from single vertex graphs with a finite sequence of the following four operations.
\begin{enumerate}
    \item Taking the \emph{disjoint union} of $G_1$ and $G_2$, defined as $(V_1 \cup V_2, E_1 \cup E_2)$.
    
    \item Taking the \emph{join} of $G_1$ and $G_2$, defined as $(V_1 \cup V_2, E_1 \cup E_2 \cup \{xy \mid x \in V_1, y \in V_2\})$.
    
    \item Adding a vertex $u \notin V_1$ \emph{comparable} to vertex $v \in V_1$, defined as $(V_1 \cup \{u\}, E_1 \cup \{ux \mid x \in X\})$, where $X \subseteq N(v)$.
    
    \item Attaching a complete graph $Q = (V_Q, E_Q)$ to a vertex $v$ of $G_1$, defined as $(V_1 \cup V_Q, E_1 \cup E_Q \cup \{qv \mid q \in V_Q\})$.
\end{enumerate}
\end{definition}

\begin{figure}
\begin{center}
\begin{tikzpicture}[scale=0.7]
\tikzstyle{vertex}=[circle, draw, fill=black, inner sep=0pt, minimum size=5pt]

	\draw (0, 0) ellipse (1cm and 1.5cm);
	\draw (2.8, 0) ellipse (1cm and 1.5cm);
	\draw[dashed] (1,0) -- (1.8, 0);
	\node at (0,0) {$G_1$};
	\node at (2.8,0) {$G_2$};
    \node at (1.4,-2) {Disjoint union};

	\draw (6, 0) ellipse (1cm and 1.5cm);
	\draw (8.8, 0) ellipse (1cm and 1.5cm);
	\draw (6.75,1) -- (8.05, -1);
	\draw (7,0) -- (7.8, 0);
	\draw (6.75,1) -- (8.05, 1);
	\draw (6.75,-1) -- (8.05, -1);
	\draw (6.75,-1) -- (8.05, 1);
	\node at (6,0) {$G_1$};
	\node at (8.8,0) {$G_2$};
	\node at (7.4,-2) {Join};
	
	\draw (13, -0.5) ellipse (1.8cm and 1cm);
	\draw (13, -0.5) ellipse (1.2cm and 0.6cm);
	\node at (13,-0.5) {\small $N(u)$};
	\node at (11.2,0.6) {\small $N(v)$};
	\node at (13,-2) {Adding a comparable};
	\node at (13,-2.7) {vertex};
	\node[vertex, label=$v$](1) at (12.1,1.3) {};
	\node[vertex, label=$u$](2) at (13.9,1.3) {};
	\draw[dashed] (12.1,1.3) -- (13.9, 1.3);
	\draw (12.1,1.3) -- (12.6, 0.48);
	\draw (13.9,1.3) -- (13.3, 0.09);
	
	\draw (18, 0) ellipse (1.5cm and 1.5cm);
	\draw (19.8, 1.7) ellipse (1cm and 1cm);
	\node[vertex, label=below:$v$](3) at (19.05,1.05) {};
    \node at (18,0) {\small $G$};
    \node at (19.8,1.7) {$Q$};
	\node at (19.4,-2) {Attaching a clique};

\end{tikzpicture}
\end{center}
\caption{The four operations that construct OAT graphs. A dashed line indicates no edge.}
\label{fig:operations}
\end{figure}
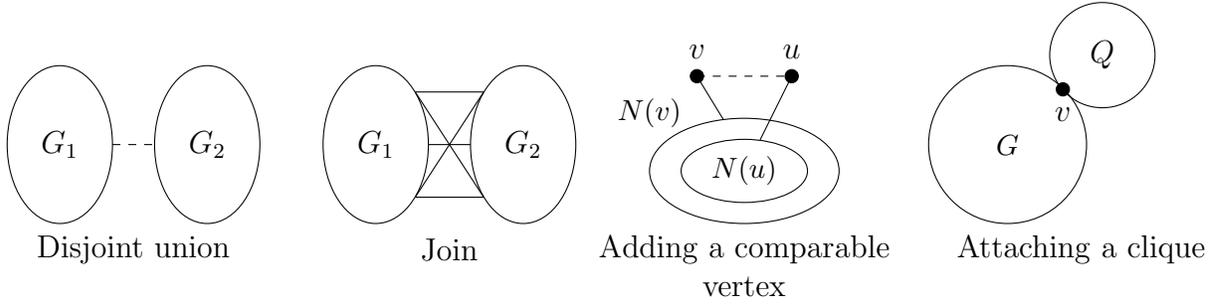

These operations are summarized in Figure \ref{fig:operations}. We will show that the class of OAT graphs contains the class of chordal bipartite graphs, compact graphs, and $P_4$-sparse graphs, a generalization of $P_4$-free graphs. For a $k$-colourable OAT graph $G$ on $n$ vertices, we prove that $\mathcal{R}_{k+1}(G)$ is connected with diameter $O(n^2)$. This result was not known for the class of $P_4$-sparse graphs. Our proof is algorithmic and gives a polynomial time algorithm to find a path between any two colourings in $\mathcal{R}_{k+1}(G)$. 

The class of OAT graphs also includes graphs which are not perfect, but all of which have equal chromatic number and clique number. Up to the knowledge of the authors, this is the first time that the $O(n^2)$ diameter of $\mathcal{R}_{k+1}(G)$ has been proven for a class that includes graphs which are not perfect (see Figure \ref{fig:imperfect}).  One can easily verify the following observation.

\begin{observation}
If $G$ is an OAT graph, then $\chi(G) = \omega(G)$.
\end{observation}

Specifically, the chromatic number and clique number of an OAT graph changes with each operation as follows.
\begin{enumerate}
    \item If $G$ is the disjoint union of the graphs $G_1$ and $G_2$, then $\chi(G) = \max\{\chi(G_1), \chi(G_2)\}$ and $\omega(G) = \max\{\omega(G_1), \omega(G_2)\}$.
    
    \item If $G$ is the join of the graphs $G_1$ and $G_2$, then $\chi(G) = \chi(G_1)+\chi(G_2)$ and $\omega(G) = \omega(G_1) + \omega(G_2)$.
    
    \item If $G$ is obtained by adding a comparable vertex to a graph $H$ then $\chi(G) = \chi(H)$ and $\omega(G) = \omega(H)$.
    
    \item If $G$ is obtained by attaching a complete graph $Q$ to a graph $H$, then $\chi(G) = \max\{\chi(H), |Q|+1\}$ and $\omega(G) = \max\{\omega(H), |Q|+1\}$.
\end{enumerate}

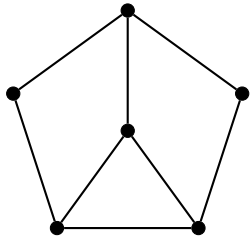
\begin{figure}
\centering
\begin{tikzpicture}[scale=0.4]
\tikzstyle{vertex}=[circle, draw, fill=black, inner sep=0pt, minimum size=5pt]
        \node[vertex](1) at (0,0) {};
        \node[vertex](7) at (0,4) {};
        \node[vertex](8) at ({4*cos(18)},{4*sin(18}) {};
	 \node[vertex](9) at ({4*sin(36)},-{4*cos(36}) {};
	 \node[vertex](10) at (-{4*sin(36)},-{4*cos(36}) {};
	 \node[vertex](11) at (-{4*cos(18)},{4*sin(18}) {};
      
    \Edge(1)(10)
 	\Edge(1)(9)
 	\Edge(1)(7)
	\Edge(7)(8)
	\Edge(8)(9)
	\Edge(9)(10)
	\Edge(10)(11)
	\Edge(7)(11)
\end{tikzpicture}
\caption{An OAT graph that is not perfect}
\label{fig:imperfect}
\end{figure}

In Section \ref{sec:preliminaries}, we provide the definitions and notations used in this paper. In Section \ref{sec:background} we discuss the relationship between OAT graphs and several other graph classes. In Section \ref{sec:recolourbackground} we discuss known results on reconfiguration of vertex colouring for these graph classes. In Section \ref{sec:recolour}, we prove our main result that if $G$ is a $k$-colourable OAT graph then $\mathcal{R}_{k+1}(G)$ is connected with diameter $O(n^2)$. In Section \ref{sec:recognize}, we show that OAT graphs can be recognized in polynomial time and we give an $O(n^3)$ time algorithm to recognize OAT graphs.

\subsection{Preliminaries}
\label{sec:preliminaries}

All graphs considered are finite and simple. Refer to the book by Diestel \cite{diestel} for definitions not given here. Let $G=(V,E)$ be a graph. The vertex-set of $G$ is denoted by $V(G)$ and the edge-set of $G$ is denoted by $E(G)$. We use $n$ to denote the number of vertices of $G$, and we use $m$ to denote the number of edges of $G$. The \emph{complement} of a graph $G$ is the graph with vertex-set $V(G)$ where two vertices are adjacent in the complement if and only if they are not adjacent in $G$. The \emph{neighbourhood} of a vertex $v \in V(G)$, denoted by $N_G(v)$ (or simply $N(v)$ if the context is clear), is the set of vertices adjacent to $v$ in $G$. The \emph{closed neighbourhood} of $v$, denoted $N[v]$, is $N(v) \cup \{v\}$. The \emph{degree} of $v$, denoted by $d(v)$, is equal to $|N(v)|$. A \emph{pendent} vertex is a vertex with degree one. A vertex $v \in V(G)$ is called a \emph{cut-vertex} if $G \setminus v$ has more components than $G$. For $X,Y\subseteq V$, we say that $X$ is \emph{joined} to $Y$ if every vertex in $X$ is adjacent to every vertex in $Y$. 

Let $P_n$ and $C_n$ denote the path and cycle on $n$ vertices, respectively. A \emph{hole} is a cycle on at least five vertices and an \emph{anti-hole} is the complement of a hole. A hole is \emph{even} or \emph{odd} if it has an even or odd number of vertices, respectively. For a graph $H$, we say that $G$ is \emph{$H$-free} if $G$ does not contain an induced subgraph isomorphic to $H$. For a set of graphs $\mathcal{H}$, we say that $G$ is $\mathcal{H}$-free if $G$ is $H$-free for every $H \in \mathcal{H}$. For convenience and when the context is clear, we consider a set of vertices $U \subseteq V(G)$ as a subgraph, namely, the subgraph of $G$ induced by $U$. We use the notation $G \setminus v$ to denote the graph obtained from $G$ by deleting $v$ and all edges incident to $v$. We use the notation $G \setminus U$ to denote the graph obtained from $G$ by deleting all vertices of $U$ and all edges incident to the vertices of $U$.

A {\em clique} (resp.~{\em stable set}) in a graph is a set of pairwise adjacent (resp.~non-adjacent) vertices. For a positive integer $k$, a \emph{$k$-colouring} of a graph $G$ is a function $c:V(G)\to \mathcal{S}$ for some finite set $\mathcal{S}$ with  $|\mathcal{S}| \le k$ and $c(u)\neq c(v)$ whenever $u$ and $v$ are adjacent in $G$. A graph is {\em $k$-colourable} if it admits a $k$-colouring. The \emph{chromatic number} of a graph $G$, denoted by $\chi(G)$, is the minimum number $k$ for which $G$ is $k$-colourable. The \emph{clique number} of $G$, denoted by $\omega(G)$, is the size of a largest clique in $G$. A graph $G$ is \emph{perfect} if for all induced subgraphs $H$ of $G$, $\chi(H) = \omega(H)$. 

For non-adjacent vertices $u,v \in V(G)$, we say that $u$ is \emph{comparable} to $v$ if $N(u) \subseteq N(v)$. For vertices $x,y \in V(G)$, if $N(x) = N(y)$ then $x$ and $y$ are called \emph{true twins} if $x$ and $y$ are adjacent and \emph{false twins} if they are not adjacent.

\section{Graph classes}
\label{sec:background}
In this section we define several classes of graphs that have been considered for reconfiguration of vertex colouring and examine the relationship between these classes (see Figure \ref{fig:classes} for a summary). A graph is \emph{chordal} if it does not contain a cycle with four or more vertices as an induced subgraph. A graph is \emph{co-chordal} if it is the complement of a chordal graph. A graph is \emph{weakly chordal} if it does not contain a hole or an anti-hole as an induced subgraph. The strong perfect graph theorem \cite{SPGT} states that a graph is perfect if and only if it does not contain a subgraph isomorphic to an odd hole or an odd anti-hole. It is clear from this characterization that every weakly chordal graph is perfect.

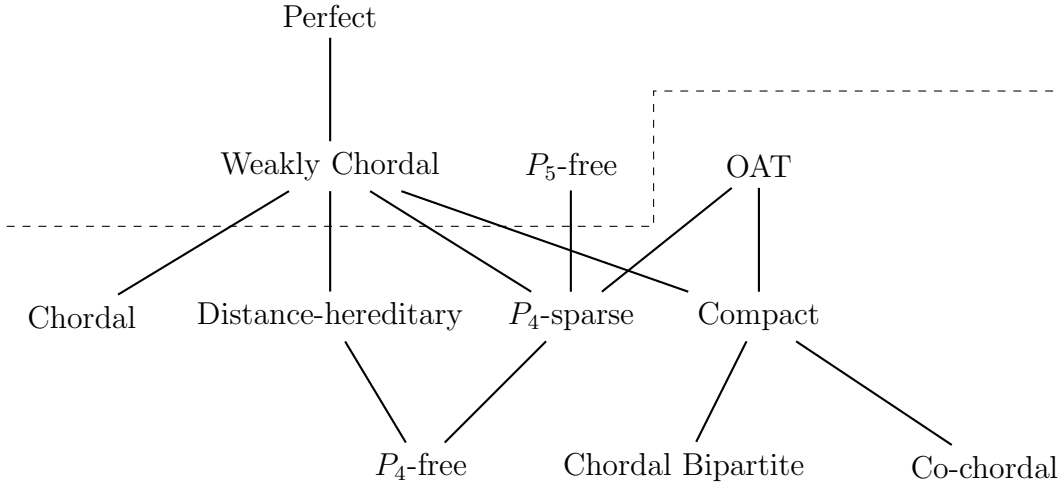
\begin{figure}[h!]
    \centering
    \begin{tikzpicture}[scale=1]
    
	\node[](1) at (-2.7,0) {Perfect};
	\node[](2) at (-2.7,-2) {Weakly Chordal};
	\node[](3) at (3,-2) {OAT};
	\node[](0) at (0.5,-2) {$P_5$-free};
	\node[](4) at (-6,-4) {Chordal};
	\node[](5) at (-2.7,-4) {Distance-hereditary};
	\node[](6) at (0.5,-4) {$P_4$-sparse};
	\node[](7) at (-1.5,-6) {$P_4$-free};
	\node[](8) at (2,-6) {Chordal Bipartite};
	\node[](9) at (3,-4) {Compact};
	\node[](10) at (6,-6) {Co-chordal};
	
	\draw[dashed] (-7,-2.8) -- (1.6,-2.8) -- (1.6,-1) -- (7,-1);
	
	\Edge(1)(2)
	\Edge(2)(4)
	\Edge(2)(5)
	\Edge(2)(6)
	\Edge(2)(9)
	\Edge(3)(6)
	\Edge(3)(9)
	\Edge(5)(7)
	\Edge(6)(7)
	\Edge(9)(10)
	\Edge(8)(9)
	\Edge(0)(6)

\end{tikzpicture}

    \caption{The relationship between graph classes. For the graph classes below the dashed line, $\mathcal{R}_{k+1}(G)$ is connected while for the graph classes above the solid line, $\mathcal{R}_{k+1}(G)$ is disconnected. The reconfiguration results for OAT graphs and $P_4$-sparse graphs are new results of this paper.}
    \label{fig:classes}
\end{figure}

A graph is \emph{bipartite} if its vertex-set can be partitioned into two stable sets. A graph is \emph{chordal bipartite} if it is both weakly chordal and bipartite. A graph $G$ is \emph{distance-hereditary} if for all connected induced subgraphs $H$ of $G$, and any two vertices $x,y$ of $H$, the distance between $x$ and $y$ in $H$ is the same as the distance between $x$ and $y$ in $G$.

The class of OAT graphs generalizes the $P_4$-free graphs, also called \emph{cographs}. The $P_4$-free graphs are exactly the graphs that can be constructed from single vertex graphs with the join and disjoint union operation \cite{corneil1981}. We also note that the class of OAT graphs extends the class of chordal bipartite graphs. Bonamy et al.~\cite{bonamy2014} proved that every chordal bipartite graph can be constructed from a set of one or more isolated vertices by adding a pendent vertex or adding a comparable vertex. It is easy to see that the operations defining OAT graphs are sufficient for building such graphs.

The class of distance-hereditary graphs is exactly the class of graphs that can be constructed from single vertex graphs by adding a pendent vertex, a true twin, or a false twin \cite{hammer1990}. It is also known that the class of distance-hereditary graphs is a generalization of the class of $P_4$-free graphs \cite{hammer1990}. Although both OAT graphs and distance-hereditary graphs extend the class of $P_4$-free graphs, both classes contain graphs that are not members of the other class. It is clear from the definition that every distance-hereditary graph does not contain a domino, house, or gem graph as an induced subgraph (see Figure \ref{fig:dhforbidden}), but each of these graphs is an OAT graph. We also note that the class of OAT graphs does not extend the class of distance-hereditary graphs, since for example, the graph in Figure \ref{fig:dhnotoat} is a distance-hereditary graph but is not an OAT graph.

\begin{figure}[h!]
\centering
\begin{tikzpicture}[scale=0.7]
\tikzstyle{vertex}=[circle, draw, fill=black, inner sep=0pt, minimum size=5pt]

    \node[vertex](1) at (0,2) {};
    \node[vertex](2) at (2,2) {};
    \node[vertex](3) at (2,4) {};
	\node[vertex](4) at (0,4) {};
	\node[vertex](5) at (0,0) {};
	\node[vertex](6) at (2,0) {};
     
    \Edge(1)(2)
    \Edge(1)(4)
    \Edge(1)(5)
    \Edge(3)(2)
    \Edge(6)(2)
    \Edge(3)(4)
    \Edge(5)(6)

    \hspace{5mm}
    
    \node[vertex](7) at (4,0) {};
    \node[vertex](8) at (6,0) {};
    \node[vertex](9) at (6,2) {};
	\node[vertex](10) at (4,2) {};
	\node[vertex](11) at (5,3.732) {};
     
    \Edge(7)(8)
    \Edge(7)(10)
    \Edge(9)(8)
    \Edge(9)(10)
    \Edge(9)(11)
    \Edge(10)(11)
    
    \hspace{5mm}
    
    \node[vertex](12) at (8,0) {};
    \node[vertex](13) at (10,0) {};
    \node[vertex](14) at (12,0) {};
	\node[vertex](15) at (14,0) {};
	\node[vertex](16) at (11,2.732) {};
     
    \Edge(12)(13)
    \Edge(13)(14)
    \Edge(14)(15)
    \Edge(15)(16)
    \Edge(14)(16)
    \Edge(13)(16)
    \Edge(12)(16)
\end{tikzpicture}
\caption{The domino, house, and gem graphs.}
\label{fig:dhforbidden}
\end{figure}
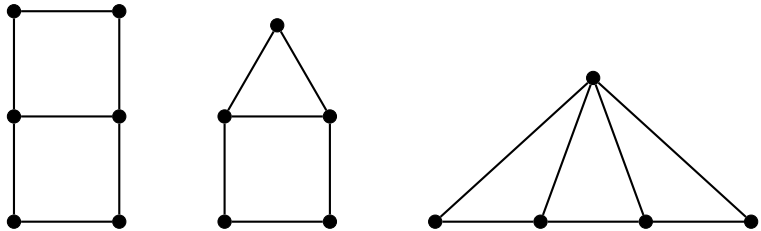

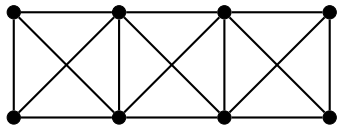
\begin{figure}[h!]
\centering
\begin{tikzpicture}[scale=0.7]
\tikzstyle{vertex}=[circle, draw, fill=black, inner sep=0pt, minimum size=5pt]

    \node[vertex](1) at (0,0) {};
    \node[vertex](2) at (0,2) {};
    \node[vertex](3) at (2,0) {};
	\node[vertex](4) at (2,2) {};
	\node[vertex](5) at (4,0) {};
	\node[vertex](6) at (4,2) {};
	\node[vertex](7) at (6,0) {};
	\node[vertex](8) at (6,2) {};
     
    \Edge(1)(2)
    \Edge(1)(3)
    \Edge(1)(4)
    \Edge(2)(3)
    \Edge(2)(4)
    \Edge(3)(4)
    \Edge(3)(5)
    \Edge(3)(6)
    \Edge(4)(5)
    \Edge(4)(6)
    \Edge(5)(6)
    \Edge(5)(7)
    \Edge(5)(8)
    \Edge(6)(7)
    \Edge(6)(8)
    \Edge(7)(8)

\end{tikzpicture}
\caption{A graph that is chordal, distance-hereditary, and not OAT.}
\label{fig:dhnotoat}
\end{figure}

Next we recall two other classes of graphs that generalize the class of $P_4$-free graphs. Jamison and Olario \cite{jamison1989} introduced the class of \emph{$P_4$-reducible} graphs and Ho\`{a}ng \cite{hoang1989} further generalized this class to the \emph{$P_4$-sparse} graphs. 

\begin{definition}[\cite{jamison1989}]
A graph is \emph{$P_4$-reducible} if each vertex of the graph is in at most one induced $P_4$.
\end{definition}

\begin{definition}[\cite{hoang1989}]
A graph is \emph{$P_4$-sparse} if for every set of 5 vertices, there is at most one induced $P_4$.
\end{definition} 

It can be verified from the definition that $P_4$-sparse graphs are ($P_5$, house, $C_5$)-free. Jamison and Olario \cite{jamison1992_1} give the complete list of seven forbidden induced subgraphs for this class. Feghali and Fiala left as an open problem if $\mathcal{R}_{k+1}(G)$ is connected with diameter $O(n^2)$ for a $k$-colourable ($P_5$, house, $C_5$)-free graph $G$. We answer this question in the positive for the subclass of $P_4$-sparse graphs. Jamison and Olario \cite{jamison1992_2} also prove that $P_4$-sparse graph are exactly the class of graphs that can be constructed from single vertex graphs with the join operation, the disjoint union operation and a third operation defined as follows (note that we only use the following operation in the proof of Lemma \ref{lem:sparse}).

Let $G_1=(V_1, \emptyset)$ and $G_2=(V_2, E_2)$ be vertex-disjoint $P_4$-sparse graphs with $V_2 = \{v\} \cup K \cup R$ such that:
\begin{itemize}
    \item $|K| = |V_1| + 1 \ge 2$.
    \item $K$ is a clique.
    \item $R$ is joined to $K$ and every vertex in $R$ is non-adjacent to $v$.
    \item There exists a vertex $v' \in K$ such that $N_{G_2}(v) = \{v'\}$ or $N_{G_2}(v) = K \setminus \{v'\}$.
\end{itemize}

Choose a bijection $f:V_1 \to K \setminus \{v'\}$. Define the third operation on $G_1$ and $G_2$ to be the graph $(V_1 \cup V_2, E_2 \cup E')$ where 

\[E' =
  \begin{cases}
    \{xf(x) \mid x \in V_1\} & \text{if $N_{G_2}(v) = \{v'\}$} \\
    \{xz \mid x \in V_1, z \in K \setminus \{f(x)\} \} & \text{if $N_{G_2}(v) = K \setminus \{v'\}$}
  \end{cases}
\]

See Figure \ref{fig:sparse} for an illustration of this operation.

\begin{lemma}
\label{lem:sparse}
Every $P_4$-sparse graph is an OAT graph.
\end{lemma}

\begin{proof}
We show that the third operation defining the class of $P_4$-sparse graphs can be replicated by repeatedly adding comparable vertices. Let $G$ be the graph constructed by this operation on $G_1$ and $G_2$. There are two cases to consider depending on whether $N_{G_2}(v) = \{v'\}$ or $N_{G_2}(v) = K \setminus \{v'\}$.

Suppose $N_{G_2}(v) = \{v'\}$. Then each vertex of $V_1$ is a pendent vertex in $G$ and is comparable to $v'$. Therefore, $G$ can be obtained from $G_2$ by adding each vertex of $G_1$ as a comparable vertex. Now suppose $N_{G_2}(v) = K \setminus \{v'\}$. Then each vertex $x \in V_1$ is comparable to $f(x)$ in $G$ since $K$ is a clique. Once again, $G$ can be obtained from $G_2$ by adding each vertex of $G_1$ as a comparable vertex.
\end{proof}

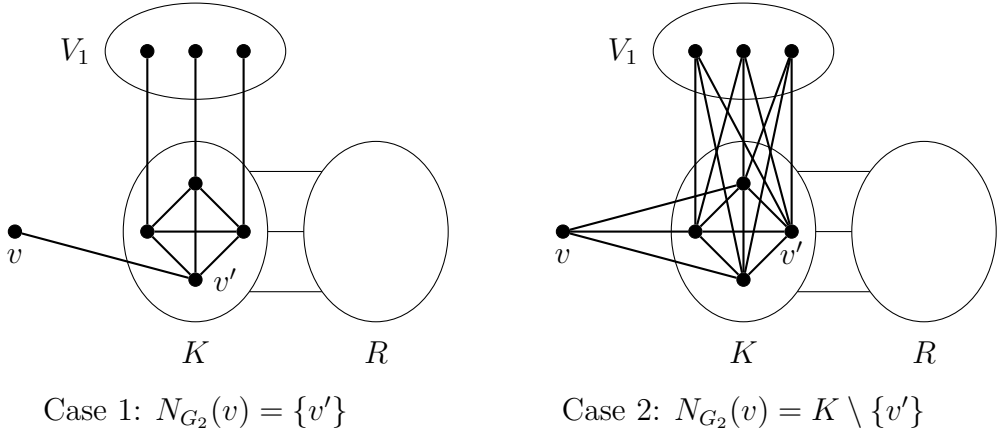
\begin{figure}
    \centering
    \begin{tikzpicture}[scale=0.8]
    \tikzstyle{vertex}=[circle, draw, fill=black, inner sep=0pt, minimum size=5pt]
        \node[vertex, label=below:$v$](1) at (0,0) {};
        \node[vertex, label=right:$v'$](0) at (3,-0.8) {};
        \node[vertex](2) at (2.2,3) {};
        \node[vertex](3) at (3,3) {};
        \node[vertex](4) at (3.8,3) {};
        \node[vertex](5) at (2.2,0) {};
        \node[vertex](6) at (3,0.8) {};
        \node[vertex](7) at (3.8,0) {};
        \draw (3, 3) ellipse (1.5cm and 0.8cm);
    	\draw (3, 0) ellipse (1.2cm and 1.5cm);
    	\draw (6, 0) ellipse (1.2cm and 1.5cm);
    	\node at (3,-2) {$K$};
    	\node at (6,-2) {$R$};
    	\node at (1,3) {$V_1$};
    	\node at (3,-3) {Case 1: $N_{G_2}(v) = \{v'\}$};
    	\draw (4.2,0) -- (4.8, 0);
	    \draw (3.9,1) -- (5.1, 1);
	    \draw (3.9,-1) -- (5.1, -1);
    	\Edge(0)(1)
    	\Edge(2)(5)
    	\Edge(3)(6)
    	\Edge(4)(7)
    	\Edge(0)(5)
    	\Edge(0)(6)
    	\Edge(0)(7)
    	\Edge(5)(6)
    	\Edge(5)(7)
    	\Edge(6)(7)
    \end{tikzpicture}
    \hspace{10mm}
    \begin{tikzpicture}[scale=0.8]
    \tikzstyle{vertex}=[circle, draw, fill=black, inner sep=0pt, minimum size=5pt]
        \node[vertex, label=below:$v$](1) at (0,0) {};
        \node[vertex](7) at (3,-0.8) {};
        \node[vertex](2) at (2.2,3) {};
        \node[vertex](3) at (3,3) {};
        \node[vertex](4) at (3.8,3) {};
        \node[vertex](5) at (2.2,0) {};
        \node[vertex](6) at (3,0.8) {};
        \node[vertex](0) at (3.8,0) {};
        \node at (3.8,-0.35) {$v'$};
        \draw (3, 3) ellipse (1.5cm and 0.8cm);
    	\draw (3, 0) ellipse (1.2cm and 1.5cm);
    	\draw (6, 0) ellipse (1.2cm and 1.5cm);
    	\node at (3,-2) {$K$};
    	\node at (6,-2) {$R$};
    	\node at (1,3) {$V_1$};
    	\node at (3,-3) {Case 2: $N_{G_2}(v) = K \setminus \{v'\}$};
    	\draw (4.2,0) -- (4.8, 0);
	    \draw (3.9,1) -- (5.1, 1);
	    \draw (3.9,-1) -- (5.1, -1);
    	
    	\Edge(1)(5)
    	\Edge(1)(6)
    	\Edge(1)(7)
    	
    	\Edge(3)(0)
    	\Edge(3)(5)
    	\Edge(3)(6)
    	
    	\Edge(2)(0)
    	\Edge(2)(5)
    	\Edge(2)(7)
    	
    	\Edge(4)(0)
    	\Edge(4)(6)
    	\Edge(4)(7)
    	
    	\Edge(0)(5)
    	\Edge(0)(6)
    	\Edge(0)(7)
    	\Edge(5)(6)
    	\Edge(5)(7)
    	\Edge(6)(7)
    \end{tikzpicture}
    \caption{The third operation defining $P_4$-sparse graphs.}
    \label{fig:sparse}
\end{figure}

Recently, Feghali and Fiala \cite{feghali2020} examined a subclass of weakly chordal graphs called \emph{compact graphs} defined below. A \emph{2-pair} $\{u,v\}$ is a pair of non-adjacent vertices such that every chordless path between $u$ and $v$ has exactly two edges. For a 2-pair $\{u,v\}$, let $S(u,v) = N(u) \cap N(v)$ and let $C_v$ denote the component of $G \setminus S(u,v)$ that contains the vertex $v$.

\begin{definition}
A weakly chordal graph $G$ is \emph{compact} if every subgraph $H$ of $G$ either
\begin{itemize}
    \item is a complete graph, or
    \item contains a 2-pair $\{x,y\}$ such that $N_H(x) \subseteq N_H(y)$, or
    \item contains a 2-pair $\{x, y\}$ such that $C_x \cup S(x,y)$ in $H$ induces a clique on at most three vertices.
\end{itemize}
\end{definition}

See also Figure \ref{fig:compact}. Feghali and Fiala proved that every co-chordal graph and every 3-colourable ($P_5$, $C_5$, house)-free graph is compact \cite{feghali2020}. We note that the class of OAT graphs is a strict generalization of compact graphs.

\begin{lemma}
Every compact graph is an OAT graph but not every OAT graph is a compact graph.
\end{lemma}
\begin{proof}
 Let $G$ be a compact graph. The proof is by induction on the number of vertices of $G$. If $G$ is a complete graph, then clearly $G$ is OAT. If not, then suppose $G$ contains a 2-pair $\{x,y\}$ such $N_G(x) \subseteq N_G(y)$. Then $x$ and $y$ are comparable vertices. By the induction hypothesis, $G \setminus x$ is an OAT graph. Then we can construct $G$ from $G \setminus x$ by adding back the comparable vertex $x$ and the appropriate edges. Lastly, suppose $G$ contains a 2-pair $\{x, y\}$ such that $C_x \cup S(x,y)$ in $H$ induces a clique on at most three vertices. If $S(x,y)$ contains two vertices then $x$ is comparable to $y$ and the proof follows from the argument above. Now assume $S(x,y)$ contains exactly one vertex $z$ and let $w$ be the vertex in $N_G(x)\setminus N_G(y)$ (see Figure \ref{fig:compact}). Then by the induction hypothesis, $G \setminus \{x, w\}$ is an OAT graph. We can construct $G$ by attaching the complete graph on two vertices $\{x, w\}$ to the vertex $z$. Therefore every compact graph is an OAT graph.
 
Next we show that there are infinitely many OAT graphs that are not compact graphs. Consider the infinite family of graphs constructed by attaching a complete graph with more than three vertices to some vertex of another complete graph with more than three vertices. By the definition of OAT graphs, it is easy to see every graph in this class is an OAT graph. It is also easy to see that any graph in this class will not satisfy any of the three requirements in the definition of compact graphs.

We consider another example showing that the class of OAT graphs contains graphs which are not compact. Note that since weakly chordal graphs are perfect, by definition compact graphs are perfect. The graph in Figure \ref{fig:imperfect} is not perfect but is an OAT graph since it can be constructed from adding three comparable vertices to a clique of size three. Therefore, not all OAT graphs are compact graphs.
\end{proof}

\begin{figure}
    \centering
    \begin{tikzpicture}[scale=0.8]
    \tikzstyle{vertex}=[circle, draw, fill=black, inner sep=0pt, minimum size=5pt]
        \node[vertex, label=above:$x$](0) at (-3,2) {};
        \node[vertex, label=above:$y$](1) at (3,2) {};
        \node[vertex, label=right:$z$](2) at (0,0) {};
        \node[vertex, label=below:$w$](3) at (-3,0) {};
        \draw (-3, 1) ellipse (0.5cm and 2cm);
    	\draw (0, 0) ellipse (1.5cm and 0.5cm);
    	\draw (3, 1) ellipse (0.5cm and 2cm);
    	\node at (-4,1) {$C_x$};
    	\node at (0,-1) {$S(x,y)$};
    	\node at (4,1) {$C_y$};
    	\node at (3,1) {\textbf{$\vdots$}};
    	\Edge(0)(2)
    	\Edge(0)(3)
    	\Edge(2)(3)
    	\Edge(2)(1)
    \end{tikzpicture}
    \caption{A compact graph that is not a complete graph and that does not have a comparable vertex.}
    \label{fig:compact}
\end{figure}
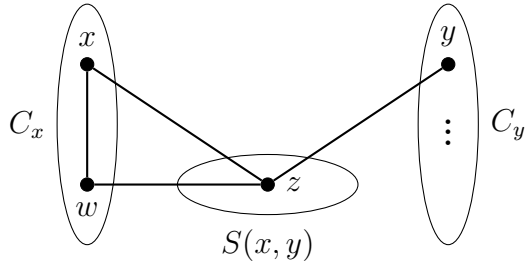

Figure \ref{fig:classes} summarizes the relationship between the graph classes discussed in this section.

\subsection{Known results on recolouring}
\label{sec:recolourbackground}

In this section we summarize reconfiguration results on graph classes that were defined in Section \ref{sec:background}. Bonamy et al.~\cite{bonamy2014} proved the following result for chordal graphs and chordal bipartite graphs.

\begin{theorem}[\cite{bonamy2014}]
Let $G$ be a $k$-colourable chordal or chordal bipartite graph on $n$ vertices. Then $\mathcal{R}_{k+1}(G)$ is connected with diameter at most $2n^2$.
\end{theorem}

In the same paper, the authors prove a lower bound on the diameter of $\mathcal{R}_{k+1}(G)$ for the path $P_n$.

\begin{theorem}[\cite{bonamy2014}]
\label{thm:lowerbound}
Let $P_n$ be the path graph on $n$ vertices. Then $\mathcal{R}_3(P_n)$ has diameter $\Omega(n^2)$.
\end{theorem}

Bonamy and Bousquet proved the following result for $P_4$-free graphs \cite{bonamy2018} and for distance-hereditary graphs \cite{bonamy20142}.

\begin{theorem}[\cite{bonamy2018}]
Let $G$ be a $k$-colourable $P_4$-free graph on $n$ vertices. Then $\mathcal{R}_{k+1}(G)$ is connected with diameter $O(\chi(G) \cdot n)$.
\end{theorem}

\begin{theorem}[\cite{bonamy20142}]
Let $G$ be a $k$-colourable distance-hereditary graph on $n$ vertices. Then $\mathcal{R}_{k+1}(G)$ is connected with diameter $O(k \cdot \chi(G) \cdot n^2)$.
\end{theorem}

Bonamy and Bousquet also mention that a similar result does not hold when the class is extended to $P_5$-free graphs. For every $k \ge 3$, they give a $k$-colourable $P_5$-free graph $G$ such that $\mathcal{R}_{2k}(G)$ has an isolated vertex. This gives motivation to study the class of $P_4$-sparse graphs which contains all $P_4$-free graphs but does not contain all $P_5$-free graphs.

Feghali and Fiala \cite{feghali2020} investigated the reconfiguration question for the class of weakly chordal graphs. They found an infinite family of $k$-colourable weakly chordal graphs $G$ where $\mathcal{R}_{k+1}(G)$ has an isolated vertex. The authors subsequently defined compact graphs as a subclass of weakly chordal graphs and proved the following result.

\begin{theorem}[\cite{feghali2020}]
Let $G$ be $k$-colourable compact graph on $n$ vertices. Then $\mathcal{R}_{k+1}(G)$ is connected with diameter at most $2n^2$.
\end{theorem}

\section{Recolouring OAT graphs}
\label{sec:recolour}

In this section we prove that for a $k$-colourable OAT graph $G$, $\mathcal{R}_{k+1}(G)$ is connected and has diameter $O(n^2)$. Our strategy uses a \emph{canonical $\chi(G)$-colouring} as a central vertex in the reconfiguration graph $\mathcal{R}_{k+1}(G)$. For any two colourings $\alpha$ and $\beta$ in $\mathcal{R}_{k+1}(G)$, we show how to transform both into the canonical $\chi(G)$-colouring $\gamma$. Then to transform $\alpha$ to $\beta$, follow the steps from $\alpha$ to $\gamma$ and then follow the steps from $\beta$ to $\gamma$ in reverse.

Let $\mathcal{S}$ be a set of $k$ colours and let $\alpha: V(G) \to \mathcal{S}$ be a $k$-colouring of $G$. We say that $\alpha$ is both a $k$-colouring and a \emph{$\mathcal{S}$-colouring} where $\mathcal{S}(\alpha) = \mathcal{S}$ denotes the set of permissible colours in $\alpha$. Let $\mathcal{C}(\alpha)$ be the set of colours $c$ such that $\alpha(v) = c$ for some vertex $v \in V(G)$. Thus $\mathcal{C}(\alpha) \subseteq \mathcal{S}(\alpha)$ but they need not be equal. We say that the colour $c$ \emph{appears} in $\alpha$ if $c \in \mathcal{C}(\alpha)$. 

We first define the reconfiguration graph of the $\mathcal{S}$-colourings. Let $\mathcal{R}_{\mathcal{S}}(G)$ be the graph whose vertices are the $\mathcal{S}$-colourings of $G$ such that two vertices of $\mathcal{R}_{\mathcal{S}}(G)$ are adjacent if and only if they differ by colour on exactly one vertex. By contrast, the definition of $\mathcal{R}_k(G)$ assumes the set of colours is $\{1, 2, \ldots, k\}$. If $|\mathcal{S}| = k$ then $\mathcal{R}_{\mathcal{S}}(G)$ is isomorphic to $\mathcal{R}_{k}(G)$. The reason behind using this more general notation has to do with the names of the colours and can be understood from the following example. Suppose $G$ is constructed by the join of $L$ and $R$. Suppose (to be concrete) that we have colourings $\alpha$ and $\beta$ of $G$ on $k=5$ colours. To reconfigure $\alpha$ to $\beta$, we will recurse on $L$ and $R$. Now, it may happen that $\alpha_L$ uses colours $\{1,4\}$ and $\beta_L$ uses colours $\{2,5\}$. We need notations to distinguish these. Our method uses two steps: to reconfigure $\alpha_L$ to a colouring that matches $\beta_L$ combinatorially, and then to rename the colours.

We say that a colouring $\alpha$ of $G$ can be \emph{transformed} into a colouring $\beta$ of $G$ in $\mathcal{R}_k(G)$ if there is a path from $\alpha$ to $\beta$ in $\mathcal{R}_k(G)$. Let $H$ be a subgraph of $G$. We use the notation $n_H$ to denote the number of vertices of $H$. The \emph{projection} of $\alpha$ onto $H$ is the colouring $\alpha_H: V(H) \to \mathcal{S}(\alpha_H)$ such that $\alpha_H(v) = \alpha(v)$ and where $\mathcal{S}(\alpha_H)$ is fixed and specified so that $\mathcal{S}(\alpha_H) \subseteq \mathcal{S}(\alpha)$. A \emph{build-sequence} of an OAT graph $G$ is a finite sequence of the four defined operations that construct $G$. A build-sequence leads to a naturally defined ordering of the vertices of $G$. For a given build-sequence, we fix this ordering of vertices to force uniqueness on the \emph{canonical $\chi(G)$-colouring of $G$}.

\begin{definition}
Let $G$ be an OAT graph and let $\mathcal{C}$ be an ordered set of $\chi(G)$ colours. Fix a build-sequence $\sigma$ of $G$. The \emph{canonical $\chi$-colouring of $G$} with respect to $\mathcal{C}$ and $\sigma$, is the $\chi(G)$-colouring of $G$ constructed recursively through $\sigma$ as follows.
\begin{enumerate}
    \item If $G$ is the disjoint union of $L$ and $R$, then take a canonical $\chi$-colouring of $L$ with respect to the first $\chi(L)$ colours of $\mathcal{C}$ and a canonical $\chi$-colouring of $R$ with respect to the first $\chi(R)$ colours of $\mathcal{C}$.
    \item If $G$ is the join of $L$ and $R$, then take a canonical $\chi$-colouring of L with respect to the first $\chi(L)$ colours in $\mathcal{C}$ and take a canonical $\chi$-colouring of $R$ with respect to the next $\chi(R)$ colours in $\mathcal{C}$.
    \item If $G$ is constructed by adding a comparable vertex $u$ to a vertex $v$ of a graph $H$, then take a canonical $\chi$-colouring of $H$ with respect to $\mathcal{C}$ and colour $u$ the same colour as $v$.
    \item If $G$ is constructed from attaching a complete graph $Q$ to a vertex $v$ of a graph $H$, take a canonical $\chi$-colouring of $H$ with respect to the first $\chi(H)$ colours of $\mathcal{C}$. Consider the vertices of $Q$ as $\{q_1, q_2, \ldots\}$. Colour the vertices $q_1, q_2, \ldots$ of $Q$ in order with the first $|Q|$ colours of $\mathcal{C}\setminus c$ where $c$ is the colour given to $v$ in the canonical $\chi$-colouring of $H$.
\end{enumerate}
\end{definition}

Note that the canonical $\chi$-colouring of $G$ with respect to $\mathcal{C}$ and $\sigma$ is unique since by induction, at each step in the construction, there is no choice on which colour a vertex is assigned. For the rest of this section, we assume the build-sequence $\sigma$ of $G$ is fixed. 

Our proofs use induction to recolour the subgraphs that build up the OAT graph in a fixed construction. There are generally two steps to these proofs. The first step is to recolour the vertices so that the partition of vertices into colour classes is the same as the target colouring, namely the canonical $\chi$-colouring. The second step is to \emph{rename} these colours so that the correct colour appears on the correct colour class. We rely on the \emph{Renaming Lemma} which states that once the vertices are partitioned into the desired colour classes, we can rename each colour class to the desired colour by recolouring each vertex at most twice.

The Renaming Lemma is an adaptation of an idea that is used in token swapping. It was discovered by Akers and Krishnamurthy \cite{akers1989}, independently by Portier and Vaughan \cite{portier1990}, and later by Pak \cite{pak1999}. It was also rediscovered by Bonamy and Bousquet \cite{bonamy2018} who rephrased the lemma in terms of recolouring complete graphs. Our statement is expressed more generally.

\begin{lemma}[Renaming Lemma \cite{bonamy2018}]
\label{lem:recolour}
If $\alpha$ and $\beta$ are two $k$-colourings of $G$ that induce the same partition of vertices into colour classes, and if $\mathcal{S}$ is a set of colours such that $\mathcal{S}(\alpha), \mathcal{S}(\beta) \subseteq \mathcal{S}$ and $|\mathcal{S}| > k$, then $\alpha$ can be transformed into $\beta$ in $\mathcal{R}_{\mathcal{S}}(G)$ by recolouring each vertex at most 2 times.
\end{lemma}

\begin{proof}
Let $V_1, V_2, \ldots, V_k$ be the partition of the vertices of $G$ into $k$ colour classes induced by both $\alpha$ and $\beta$. Let $Q$ be the complete graph on $k$ vertices $\{q_1, q_2, \ldots q_k\}$ and let $\alpha_Q$ (resp. $\beta_Q$) be the colouring of $Q$ where $q_i$ is coloured the same as each vertex of $V_i$ in $\alpha$ (resp. $\beta$) for all $i=1\ldots k$. Suppose that $\alpha_Q$ can be transformed into $\beta_Q$ in $\mathcal{R}_{\mathcal{S}}(Q)$. Then to transform $\alpha$ into $\beta$ in $\mathcal{R}_{\mathcal{S}}(G)$, follow the steps from $\alpha_Q$ to $\beta_Q$. Whenever $q_i$ is recoloured, then recolour every vertex in $V_i$ the same colour. Therefore, it is enough to show how to transform $\alpha_Q$ into $\beta_Q$ in $\mathcal{R}_{\mathcal{S}}(Q)$. 

Let $D$ be the directed graph on $k$ vertices such that there is an arc $q_jq_i$ in $D$ if and only if in the current colouring of $Q$, $q_j$ is coloured $\beta(q_i)$. Since no two vertices of $Q$ are coloured the same colour in any colouring, $d^-(q_i) \le 1$ and $d^+(q_i) \le 1$ for all $i=1 \ldots k$. Therefore, $D$ is the disjoint union of directed paths and directed cycles. Note that for any vertex $q$, if $d^-(q)=0$ in $D$, then it can be immediately recoloured into $\beta_Q(q)$. 

Recolour each directed path $v_1, v_2, \ldots, v_p$ as follows. Since $d^-(v_1)=0$ recolour it $\beta_Q(v_1)$. Now we have that $d^-(v_2)=0$ so recolour $v_{2}$ with $\beta_Q(v_2)$. Continue recolouring this way until all vertices in the path are coloured as in $\beta_Q$. Note that each vertex in the directed path was recoloured at most once.

Recolour each directed cycle $v_1, v_2, \ldots, v_p, v_1$ as follows. Since $|\mathcal{S}| > k$, $v_p$ can be recoloured some colour that does not appear in the current colouring. Now $d^-(v_1)=0$ so the directed cycle becomes a directed path $v_1, v_2, \ldots, v_p$. Recolour each vertex as described in the case of a directed path.

Note that each vertex in a directed path was recoloured at most once. Also only one vertex in each directed cycle was recoloured at most twice, and each other vertex in the cycle was recoloured at most once.
\end{proof}

Before stating the next lemma, we note a lower bound on the diameter of $\mathcal{R}_{k+1}(G)$. The path $P_n$ for all $n \ge 1$ is an OAT graph with $\chi(P_n) = 2$.  Bonamy et al.~\cite{bonamy2014} proved that $\mathcal{R}_3(P_n)$ has diameter $\Omega(n^2)$ (see Theorem \ref{thm:lowerbound}). Thus for a general $k$-colourable OAT graph $G$, the diameter of $\mathcal{R}_{k+1}(G)$ is $\Omega(n^2)$.

\begin{lemma}
\label{lem:main1}
Let $G$ be a $k$-colourable OAT graph. Let $\mathcal{S}$ be a set of $k+1$ colours and let $\mathcal{C}$ be an ordered set of $\chi(G)$ colours such that $\mathcal{C} \subseteq \mathcal{S}$. Then any colouring in $\mathcal{R}_{\mathcal{S}}(G)$ can be transformed into the canonical $\chi$-colouring of $G$ with respect to $\mathcal{C}$ by recolouring each vertex at most $2n$ times.
\end{lemma}

\begin{proof}
The proof is by induction on the number of vertices $n$ of $G$. Let $\alpha: V(G) \to \mathcal{S}$ be a $(k+1)$-colouring of $G$. We show how to transform $\alpha$ into the canonical $\chi$-colouring $\gamma$ of $G$ with respect to $\mathcal{C}$ by recolouring each vertex at most $2n$ times.

\begin{case}
Suppose $G$ is constructed by the disjoint union of the graphs $L$ and $R$. Note that $L$ and $R$ can be recoloured independently since there are no edges between $L$ and $R$. Since $\chi(G) = \max\{\chi(L), \chi(R)\}$, it follows that $|\mathcal{S}| \ge \chi(L) + 1$ and $|\mathcal{S}| \ge \chi(R) + 1$. Let $\alpha_L$ be the projection of $\alpha$ onto $L$ and let $\mathcal{S}(\alpha_L) = \mathcal{S}$ be the set of permissible colours. Then $\alpha_L$ is a $(k_L+1)$-colouring of $L$ for some $k_L \ge \chi(L)$ in $\mathcal{R}_{\mathcal{S}}(L)$. Let $\alpha_R$ be the projection of $\alpha$ onto $R$ and let $\mathcal{S}(\alpha_R) = \mathcal{S}$ be the set of permissible colours. Then $\alpha_R$ is a $(k_R+1)$-colouring of $R$ for some $k_R \ge \chi(R)$ in $\mathcal{R}_{\mathcal{S}}(R)$. By the induction hypothesis, we can transform $\alpha_L$ into the canonical $\chi$-colouring of $L$ with respect to the first $\chi(L)$ colours of $\mathcal{C}$ by recolouring each vertex of $L$ at most $2n_L$ times. Similarly, by the induction hypothesis, we can transform $\alpha_R$ into the canonical $\chi$-colouring of $R$ with respect to the first $\chi(R)$ colours of $\mathcal{C}$ by recolouring each vertex of $R$ at most $2n_R$ times. Taking these two sequences together gives the recolouring sequence to transform $\alpha$ into the canonical $\chi$-colouring of $G$ with respect to $\mathcal{C}$. Each vertex of $L$ has been recoloured at most $2n_L < 2n$ times and each vertex of $R$ has been recoloured at most $2n_R < 2n$ times. 
\end{case}

\begin{case}
Suppose $G$ is constructed by the join of the graphs $L$ and $R$. Note that $\mathcal{C}(\alpha_L)$ is disjoint from $\mathcal{C}(\alpha_R)$ since there are all possible edges between $L$ and $R$. We consider two cases depending on the number of colours appearing in $\alpha_L$ and $\alpha_R$.

First suppose $|\mathcal{C}(\alpha_L)| = \chi(L)$ and $|\mathcal{C}(\alpha_R)| = \chi(R)$. Then there exists some colour $c$ that does not appear in $\alpha$ since $\chi(G) = \chi(L)+\chi(R)$ and $|\mathcal{S}| > \chi(G)$. Let $\alpha_L$ be the projection of $\alpha$ onto $L$ and let $\mathcal{S}(\alpha_L) = \mathcal{C}(\alpha_L) \cup \{c\}$ be the set of permissible colours. Then $\alpha_L$ is a $(k_L+1)$-colouring for some $k_L \ge \chi(L)$ in $\mathcal{R}_{\mathcal{S}}(L)$. By the induction hypothesis, $\alpha_L$ can be transformed into the canonical $\chi$-colouring of $L$ with respect to the first $\chi(L)$ colours of $\mathcal{C}(\alpha_L)$ by recolouring each vertex at most $2n_L$ times. Now again some colour $c'$ does not appear in the current colouring of $G$. Let $\alpha_R$ be the projection of $\alpha$ onto $R$ and let $\mathcal{S}(\alpha_R) = \mathcal{C}(\alpha_R) \cup \{c'\}$ be the set of permissible colours. Then $\alpha_R$ is a $(k_R+1)$-colouring for some $k_R \ge \chi(R)$ in $\mathcal{R}_{\mathcal{S}}(R)$. By the induction hypothesis, $\alpha_R$ can be transformed into the canonical $\chi$-colouring of $R$ with respect the first $\chi(R)$ colours of $\mathcal{C}(\alpha_R)$ by recolouring each vertex at most $2n_R$ times. 

Now suppose $|\mathcal{C}(\alpha_L)| > \chi(L)$ or $|\mathcal{C}(\alpha_R)| > \chi(R)$ (suppose the former). Let $\alpha_L$ be the projection of $\alpha$ onto $L$ and let $\mathcal{S}(\alpha_L)=\mathcal{C}(\alpha_L)$ be the set of permissible colours. Then $\alpha_L$ is a $(k_L+1)$-colouring for some $k_L \ge \chi(L)$ in $\mathcal{R}_{\mathcal{S}(\alpha_L)}(L)$. By the induction hypothesis, $\alpha_L$ can be transformed into the canonical $\chi$-colouring of $L$ with respect the first $\chi(L)$ colours of $\mathcal{C}(\alpha_L)$ by recolouring each vertex at most $2n_L$ times. Now some colour $c^*$ that appeared in $\alpha_L$ no longer appears in the current colouring of $G$. Let $\alpha_R$ be the projection of $\alpha$ onto $R$ and let $\mathcal{S}(\alpha_R) = \mathcal{C}(\alpha_R) \cup \{c^*\}$ be the set of permissible colours. Then $\alpha_R$ is a $(k_R+1)$-colouring for some $k_R \ge \chi(R)$ in $\mathcal{R}_{\mathcal{S}(\alpha_R)}(R)$. By the induction hypothesis, $\alpha_R$ can be transformed into the canonical $\chi$-colouring of $R$ with respect the first $\chi(R)$ colours of $\mathcal{C}(\alpha_R)$ by recolouring each vertex at most $2n_R$ times. A similar argument holds if instead $|\mathcal{C}(\alpha_R)| > \chi(R)$.

To complete this part of the proof, we now have a colouring $\alpha'$ of $G$ such that $\alpha'_L$ is a canonical $\chi(L)$-colouring of $L$ and $\alpha'_R$ is a canonical $\chi(R)$-colouring of $R$. Then $\alpha'$ and the canonical $\chi(G)$-colouring $\gamma$ of $G$ must partition the vertices of $G$ into the same colour classes. Then by the Renaming Lemma (Lemma \ref{lem:recolour}), we can transform $\alpha'$ into $\gamma$ by recolouring each vertex at most twice. Therefore we can transform $\alpha$ into $\gamma$ by recolouring each vertex of $G$ at most $2\max\{n_L,n_R\} + 2 \le 2n$ times.
\end{case}

\begin{case}
Suppose $G$ is constructed by adding a vertex $v^*$ comparable to a vertex $v$ of the OAT graph $H = G \setminus \{v^*\}$. First recolour $v^*$ the same colour as $v$. This is possible since $v^*$ and $v$ are non-adjacent and $N(v^*) \subseteq N(v)$. Let $\alpha_H$ be the projection of $\alpha$ onto $H$ and let $\mathcal{S}(\alpha_H) = \mathcal{S}(\alpha)$ be the set of permissible colours. Note that since $\chi(H) = \chi(G)$, then $\alpha_H$ is a $(k+1)$-colouring in $\mathcal{R}_{\mathcal{S}(\alpha_H)}(H)$. By the induction hypothesis, $\alpha_H$ can be transformed into the canonical $\chi(H)$-colouring with respect to $\mathcal{C}$ by recolouring each vertex of $H$ at most $2n_H$ times. To extend this recolouring sequence to $G$, whenever $v$ is recoloured, recolour $v^*$ the same colour. By definition, this is the canonical $\chi$-colouring of $G$ with respect to $\mathcal{C}$. Each vertex of $H$ was recoloured at most $2n_H < 2n$ times and $v^*$ was recoloured at most $2n_H+1 < 2n$ times.
\end{case}

\begin{case}
Suppose $G$ is constructed by attaching a complete graph $Q$ to some vertex $v$ of a graph $H$. Let $\alpha_H$ be the projection of $\alpha$ onto $H$ and let $\mathcal{S}(\alpha_H) = \mathcal{S}$ be the set of permissible colours. Then $\alpha_H$ is a $(k_H + 1)$-colouring for some $k_H \ge \chi(H)$ in $\mathcal{R}_{\mathcal{S}(\alpha_H)}(H)$. By the induction hypothesis, $\alpha_H$ can be transformed into the canonical $\chi$-colouring $\gamma_H$ of $H$ with respect to the first $\chi(H)$-colours of $\mathcal{C}$. To extend this recolouring sequence to the entire graph, whenever $v$ is recoloured to some colour $c$, we may need to first recolour at most one vertex $q$ of $Q$ that is coloured $c$. Since $\chi(G) = \max\{\chi(H), n_Q+1\}$ and $|\mathcal{S}| \ge \chi(G)+1 \ge n_Q+2$, and each vertex of $Q$ has degree $n_Q$, there exists some colour $c'$ that does not appear on the neighbourhood of $q$ and is not the colour $c$. Recolour $q$ with the colour $c'$ and then continue by recolouring $v$ colour $c$. Now $H$ is coloured with the canonical $\chi(H)$-colouring $\gamma_H$. 

Let $c^* = \gamma_H(v)$ and let $\alpha_Q'$ be the current colouring of $Q$ with $\mathcal{S}(\alpha_Q') = \mathcal{S} \setminus \{c^*\}$ as the set of permissible colours. Recall that the vertices of $Q$ are given by $\{q_1, q_2, \ldots \}$. The canonical $\chi$-colouring of $Q$ with respect to $\mathcal{C}$ is the colouring $\gamma_Q$ such that $q_i$ is coloured the $i$th colour of $\mathcal{C} \setminus \{c^*\}$. Since $|\mathcal{S}| \ge n_Q + 2$, then $|\mathcal{S} \setminus \{c^*\}| \ge n_Q + 1$. By the Renaming Lemma (Lemma \ref{lem:recolour}), $\alpha_Q'$ can be transformed into $\gamma_Q$ by recolouring each vertex of $Q$ at most twice. Since each vertex of $Q$ is only adjacent to $v$ in $H$ and $c^*$ was never used in this recolouring of $Q$, this recolouring sequence can extend to $G$. Now by definition, the current colouring of $G$ is the canonical $\chi$-colouring of $G$ with respect to $\mathcal{C}$. Each vertex of $H$ was recoloured at most $2n_H$ times and each vertex of $Q$ was recoloured at most $2n_H + 2 \le 2n$ times.
\end{case}
\end{proof}

\begin{theorem}
\label{thm:main}
Let $G$ be a $k$-colourable OAT graph. Then $\mathcal{R}_{k+1}(G)$ is connected with diameter at most $4n^2$.
\end{theorem}

\begin{proof}
Fix $\mathcal{S} = \{1, 2, \ldots, k+1\}$ to be the set of permissible colours used in the colourings of $\mathcal{R}_{k+1}(G)$ and let $\mathcal{C}$ be an ordered set of $\chi(G)$ colours such that $\mathcal{C} \subseteq \mathcal{S}$. Let $\alpha,\beta:V(G) \to \mathcal{S}$ be two $(k+1)$-colourings of $G$. Then by Lemma \ref{lem:main1}, we can transform both $\alpha$ and $\beta$ into the canonical $\chi$-colouring of $G$ with respect to $\mathcal{C}$ in $\mathcal{R}_{\mathcal{S}}(G)$ by recolouring each vertex at most $2n$ times. Therefore $\mathcal{R}_{k+1}(G)$ is connected with diameter at most $4n^2$.
\end{proof}

\section{Recognizing OAT graphs}
\label{sec:recognize}

In this section we show that OAT graphs can be recognized in $O(n^3)$ time. We first survey recognition algorithms for various graph classes discussed in this paper. The following graph classes have linear time recognition algorithms: chordal graphs \cite{rose1976}, co-chordal graphs \cite{hoang1990}, $P_4$-free graphs \cite{corneil1985}, $P_4$-sparse graphs \cite{jamison1992_2}, and distance-hereditary graphs \cite{hammer1990}. Chordal bipartite graphs can be recognized in $O(\min\{n^2, (n+m)\log{n}\})$ time \cite{lubiw1987, paige1987, spinrad1993},  weakly chordal graphs can be recognized in $O(mn^2)$ time \cite{spinrad1995}, and perfect graphs can be recognized in $O(n^9)$ time \cite{chudnovsky2005}. 

Next we prove several useful lemmas that together imply that OAT graphs can be recognized by greedily deconstructing them.

\begin{lemma}
\label{lem:op1}
A graph $G$ is an OAT graph if and only if every connected component of $G$ is an OAT graph.
\end{lemma}

\begin{proof}
Suppose every component of $G$ is an OAT graph. Then the disjoint union operation can be used repeatedly to construct $G$. Now suppose $G$ is an OAT graph and that $G$ has more than one component. The proof is by induction on the number of vertices of $G$. By definition, $G$ was constructed using the operations described in the definition of OAT graphs. Then $G$ was not constructed from the join operation since $G$ is disconnected. The other operations preserve the connected components of $G$. In some construction of $G$, the disjoint union operation can be swapped with these operations.

\begin{case1}
Suppose $G$ was constructed by the disjoint union operation of $L$ and $R$. By definition, $L$ and $R$ are OAT graphs. By the induction hypothesis, every component of $L$ is an OAT graph and every component of $R$ is an OAT graph. Note that the components of $G$ are just the components of $L$ and $R$. Therefore, every component of $G$ is an OAT graph as desired.
\end{case1}

\begin{case1}
Suppose $G$ was constructed by adding a vertex $u$ comparable to a vertex $v$ of $H$. By definition $H$ is an OAT graph. By the induction hypothesis, every component of $H$ is an OAT graph. Clearly, $v$ only has neighbours in one component $H^*$ of $H$. Then $G[H^* \cup u]$ is an OAT graph since it was constructed from an OAT graph by adding a comparable vertex. Since all other components of $G$ are the components of $H$, every component of $G$ is an OAT graph.
\end{case1}

\begin{case1}
Suppose $G$ was constructed by attaching a clique $Q$ to a vertex $z$ of a graph $H$. By definition, $H$ is an OAT graph. By the induction hypothesis, each component of $H$ is an OAT graph. Suppose $z$ is a vertex of a component $H^*$ of $H$. Then $G[H^* \cup Q]$ is an OAT graph since it was constructed from an OAT graph by attaching a clique. Since all other components of $G$ are the components of $H$, every component of $G$ is an OAT graph.
\end{case1}
\end{proof}

\begin{lemma}
\label{lem:op2}
Suppose the vertices of $G$ can be partitioned into two sets $L$ and $R$ such that $L$ is joined to $R$. Then $G$ is an OAT graph if and only if $L$ is an OAT graph and $R$ is an OAT graph.
\end{lemma}

\begin{proof}
Suppose $L$ and $R$ are both OAT graphs. Then $G$ is an OAT graph since $G$ can be constructed from the join of $L$ and $R$. Now suppose that $G$ is an OAT graph. The proof is by induction on the number of vertices of $G$. By definition, $G$ was constructed by the operations described in the definition of OAT graphs. Note that $G$ was not constructed by the disjoint union operation since $G$ is connected. 

\begin{case2}
Suppose $G$ is constructed from the join of $L^*$ and $R^*$. Assume that $L \neq L^*$ (and $R \neq R^*$) and $L \neq R^*$ (and $R \neq L^*$) since otherwise we are done. By definition $L^*$ and $R^*$ are OAT graphs. Then $L$ can be partitioned into two sets $L_1 = L^* \cap L$ and $L_2 = R^* \cap L$. Similarly, $R$ can be partitioned into two sets $R_1 = L^* \cap R$ and $R_2 = R^* \cap R$ (see Figure \ref{fig:joins}). Since $L^*$ is joined to $R^*$, $L_1$ is joined to $L_2$ and $R_1$ is joined to $R_2$. Since $L^* = L_1 \cup R_1$ is an OAT graph and $L_1$ is joined to $R_1$, by the induction hypothesis, $L_1$ and $R_1$ are OAT graphs. Similarly, since $R^* = L_2 \cup R_2$ is an OAT graph and $L_2$ is joined to $R_2$, by the induction hypothesis, $L_2$ and $R_2$ are OAT graphs. Since $L$ is the join of $L_1$ and $L_2$, $L$ is an OAT graph. Similarly, since $R$ is the join of $R_1$ and $R_2$, $R$ is an OAT graph.
\end{case2}

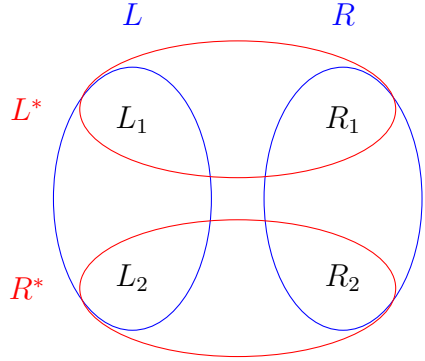
\begin{figure}
\begin{center}
\begin{tikzpicture}[scale=0.7]
\tikzstyle{vertex}=[circle, draw, fill=black, inner sep=0pt, minimum size=5pt]

	\draw[blue] (6, 0) ellipse (1.5cm and 2.5cm);
	\draw[blue] (10, 0) ellipse (1.5cm and 2.5cm);
	\draw[red] (8, 1.7) ellipse (3cm and 1.3cm);
	\draw[red] (8, -1.7) ellipse (3cm and 1.3cm);
	\node[blue] at (6,3.5) {$L$};
	\node[blue] at (10,3.5) {$R$};
	\node[red] at (4,1.7) {$L^*$};
	\node[red] at (4,-1.7) {$R^*$};
	\node at (6,1.5) {$L_1$};
	\node at (6,-1.5) {$L_2$};
	\node at (10,1.5) {$R_1$};
	\node at (10,-1.5) {$R_2$};
	
	
	
\end{tikzpicture}
\end{center}
\caption{Two partitions of $G$ into sets that are joined to each other.}
\label{fig:joins}
\end{figure}

\begin{case2}
Suppose $G$ was constructed by adding a vertex $u$ comparable to a vertex $v$ of $H$. Then $u$ and $v$ must both be in $L$ or both be in $R$ since $u$ and $v$ are non-adjacent. Without loss of generality, suppose $u, v \in L$. By definition, $H$ is an OAT graph. Note that $H$ is the join of $L \setminus u$ and $R$. Then by the induction hypothesis $L \setminus u$ and $R$ are OAT graphs. Then $L$ is an OAT graph since it can be constructed from $L \setminus u$ by adding the vertex $u$ comparable to $v$ and the appropriate edges. 
\end{case2}

\begin{case2}
Suppose $G$ was constructed by attaching a clique $Q$ to a vertex $z$ of a graph $H$. Since there are no edges between $Q$ and $H \setminus z$ in $G$, either the vertices of both $Q$ and $H \setminus z$ are in $R$ or are in $L$. Without loss of generality, suppose the vertices of $Q$ and $H \setminus z$ are in $L$. Then since $R \neq \emptyset$, $z$ is the only vertex in $R$. Then clearly $R$ is an OAT graph. Note that $L$ is composed of the graphs $H \setminus z$ and $Q$. By definition, $H$ is an OAT graph. Note that since $z \in R$, $H$ is the join of $H \setminus z$ and $z$. Then by the induction hypothesis, $H \setminus z$ is an OAT graph. Then $L$ is an OAT graph since $L$ can be constructed by the disjoint union of $H \setminus z$ and $Q$.
\end{case2}
\end{proof}

\begin{lemma}
\label{lem:op3}
If $G$ has a vertex $u$ comparable to some vertex $v$ of $G$, then $G$ is an OAT graph if and only if $G \setminus u$ is an OAT graph.
\end{lemma}

\begin{proof}
Suppose $G \setminus u$ is an OAT graph. Then $G$ is an OAT graph since $G$ can be constructed from $G \setminus u$ by adding back the comparable vertex $u$ and all appropriate edges. Now suppose $G$ is an OAT graph. The proof is by induction on the number of vertices of $G$. By definition, $G$ was constructed by using the operations described in the definition of OAT graphs.

\begin{case3}
Suppose $G$ was constructed from the disjoint union of the OAT graphs $L$ and $R$. Without loss of generality, suppose $u \in L$. Assume that $L$ has another vertex, since if not $G \setminus u$ is OAT by definition. If $u$ is an isolated vertex, i.e. $N_G(u) = \emptyset$, then $u$ is comparable to every other vertex of $L$. Otherwise, both $u$ and $v$ must be in $L$ since there is a path connecting them through $N(u) \cap N(v)$. In either case, $u$ is comparable to some vertex of $L$. By the induction hypothesis, $L \setminus u$ is an OAT graph. Then $G \setminus u$ is an OAT graph since $G$ can be constructed from the disjoint union of $L \setminus u$ and $R$.
\end{case3}

\begin{case3}
Suppose $G$ was constructed from the join of the OAT graphs $L$ and $R$. Then both $u$ and $v$ are vertices of $L$ or both $u$ and $v$ are vertices of $R$, since there are all possible edges between $L$ and $R$. Without loss of generality, suppose $u$ and $v$ are vertices of $L$. Since $u$ and $v$ are comparable in $G$, then removing the same set of vertices from $N(u)$ and $N(v)$ will not change that $N(u) \subseteq N(v)$, and so $u$ is comparable to $v$ in $L$. By the induction hypothesis, $L\setminus u$ is an OAT graph. Then $G \setminus u$ is an OAT graph since it can be constructed from the join of $L\setminus u$ and $R$.
\end{case3}

\begin{case3}
Suppose $G$ was constructed by adding a vertex $x$ comparable to another vertex $y$ of an OAT graph $H$. Assume $u \neq x$ since otherwise $G \setminus u$ is an OAT graph by definition. Then since $G \setminus x$ is an OAT graph, by the induction hypothesis, $G \setminus \{x, u\}$ is an OAT graph. Next we show that $x$ is comparable to some vertex of $G \setminus \{x, u\}$. If $u \neq y$, then $y$ is a vertex of $G \setminus \{x, u\}$ and therefore $x$ is comparable to some vertex of $G \setminus \{x, u\}$. If $u = y$, then $v$ must be a vertex of $G \setminus \{x, u\}$ and $N(x) \subseteq N(y) =  N(u) \subseteq N(v)$. Furthermore, $x$ is non-adjacent to $v$ in $G$ since $u$ is non-adjacent to $v$ and $N(x) \subseteq N(u)$. In any case, $x$ is comparable to some vertex of $G \setminus \{x, u\}$ and we can construct $G \setminus u$ from $G \setminus \{x, u\}$ by adding back $x$ and the appropriate edges. Therefore, $G \setminus u$ is an OAT graph.
\end{case3}

\begin{case3}
Suppose $G$ was constructed by attaching a complete graph $Q$ to a vertex $z$ of an OAT graph $H$. Suppose that $H$ has a vertex other than $z$, since otherwise $G$ is a clique and clearly $G \setminus u$ is an OAT graph. Note that $u \neq z$ since $z$ is adjacent to every vertex of $Q$ and no vertex of $H \setminus z$ is adjacent to a vertex of $Q$. Suppose $u$ is a vertex of $Q$. Note that $u$ is not comparable to $z$ or any other vertex in $Q$ since $u$ is adjacent to $z$ and all other vertices in $Q$. Then $u$ is the only vertex of $Q$ since no vertex of $H$ is adjacent to a vertex of $Q$. Then $G \setminus u = H$ and therefore $G \setminus u$ is an OAT graph. Now assume $u \in H \setminus z$. By the induction hypothesis, $H \setminus u$ is an OAT graph. Then $G \setminus u$ is an OAT graph since $G$ can be constructed by attaching $Q$ to the vertex $z$ of $H \setminus u$.
\end{case3}
\end{proof}

\begin{lemma}
\label{lem:op4}
Suppose $G$ has a cut vertex $z$ such that $G \setminus z$ has a component $Q$ that is a complete graph and such that $z$ is joined to $Q$. Then $G$ is an OAT graph if and only if $G \setminus Q$ is an OAT graph.
\end{lemma}

\begin{proof}
Let $H = G \setminus Q$ and suppose $H$ is an OAT graph. Then $G$ is an OAT graph since $G$ can be constructed from $H$ by attaching a clique $Q$ to the vertex $z$. Now suppose $G$ is an OAT graph. The proof is by induction on the number of vertices of $G$. By definition, $G$ can be constructed by using the operations described in the definition of OAT graphs.

\begin{case4}
Suppose $G$ was constructed from the disjoint union of the OAT graphs $L$ and $R$. Without loss of generality, suppose $z \in L$. Then every vertex of $Q$ is in $L$ since $z$ is adjacent to every vertex of $Q$. By the induction hypothesis, $L \setminus Q$ is an OAT graph. Then $G \setminus Q$ is an OAT graph since $G$ can be constructed from the disjoint union of $L \setminus Q$ and $R$.
\end{case4}

\begin{case4}
Suppose $G$ was constructed from the join of the OAT graphs $L$ and $R$. Then all of the vertices of $Q$ and $H$ are in $R$ or in $L$ since there are no edges between $Q$ and $H$. Without loss of generality, suppose all of the vertices of $Q$ and $H$ are in $L$. Then since $R$ is not empty, it must be the case that $z \in R$. By the induction hypothesis, $L \setminus Q$ is an OAT graph. Then $G$ is an OAT graph since $G$ can be constructed by the join of $L \setminus Q$ and $R$.
\end{case4}

\begin{case4}
Suppose $G$ was constructed by adding a vertex $x$ comparable to another vertex $y$ of an OAT graph $H$. Note that $z \neq x$ since $z$ is adjacent to every vertex of $Q$ and no other vertex of $G \setminus Q$ is adjacent to a vertex of $Q$. If $x \in Q$ then $x$ is the only vertex of $Q$ since $x$ is adjacent to every other vertex of $Q$ and $z$ and no other vertex of $G \setminus Q$ is adjacent to a vertex in $Q$. Then $H = G \setminus Q$ and therefore $G \setminus Q$ is an OAT graph. Now assume $x \in G \setminus Q$ and $x \neq z$. Let $G' = G \setminus x$. By the induction hypothesis, $G' \setminus Q$ is an OAT graph. If $y \in G \setminus Q$, then $G \setminus Q$ is an OAT graph since $G \setminus Q$ can be constructed from $G' \setminus Q$ by adding back the vertex $x$ comparable to $y$. If $y \in Q$, then either $x$ is an isolated vertex or $x$ is only adjacent to $z$ in $G$. If $x$ is an isolated vertex, then $G \setminus Q$ is an OAT graph since it can be constructed from the disjoint union of $x$ and $G' \setminus Q$. If $x$ is only adjacent to $z$ in $G$, then $G \setminus Q$ is an OAT graph since it can be constructed by attaching the clique $x$ to the vertex $z$ of $G' \setminus Q$.
\end{case4}

\begin{case4}
Suppose $G$ was constructed by attaching a complete graph $Q'$ to a vertex $z'$ of an OAT graph $H'$. Assume that $Q \neq Q'$ (and therefore disjoint) since if not, we are done. Then by the induction hypothesis, $H' \setminus Q$ is an OAT graph. Then $G \setminus Q$ is an OAT graph since it can be constructed from $H \setminus Q$ by attaching the clique $Q'$ at the vertex $z'$.
\end{case4}

\end{proof}

\subsection{Recognition Algorithm}

In this section we give an algorithm to recognize OAT graphs in $O(n^3)$ time. The algorithm takes a graph $G$ as input and either outputs a \emph{build-tree} certifying that $G$ is an OAT graph or outputs the answer no. A \emph{build-tree} is a tree whose root node is $G$, whose leaf nodes are the single vertices of $G$, and whose internal nodes represent the operations used to build $G$. The internal nodes have exactly two children which are the input graphs for the corresponding operation. We note that a build-tree of $G$ is not unique, and the algorithm only returns one of possibly many build-trees. For example, take a graph constructed by attaching a complete graph to a vertex of another complete graph. This graph could also be constructed by joining a single vertex to the disjoint union of two complete graphs.

Let $G$ be the input graph. The idea of the algorithm is to check if any of the four defining operations can be used to construct $G$ from smaller graphs. If so, then we recursively test those smaller graphs. Lemma \ref{lem:op1}, \ref{lem:op2}, \ref{lem:op3}, and \ref{lem:op4} justify this approach. In particular, it does not matter in which order we check the four operations. Here are further details that (arbitrarily) use the order of operations from the definition.

The algorithm first checks if $G$ is connected. If not, the algorithm recursively checks each connected component of $G$. If each component is an OAT graph, then $G$ is an OAT graph and can be constructed by the disjoint union of its components and the build-tree is updated accordingly. If $G$ is connected, then the algorithm checks if the complement of $G$ is connected. If not, the algorithm recursively checks the complement of the connected components of the complement of $G$. If each of these graphs are OAT graphs, then $G$ is an OAT graph and can be constructed by the join of these graphs, and the build-tree is updated accordingly. Next if $G$ is connected and the complement of $G$ is connected, the algorithm examines if $G$ contains a vertex $u$ comparable to a vertex $v$. If so, then the algorithm recursively checks if $G \setminus u$ is an OAT graph. If so, then $G$ is an OAT graph and can be constructed from $G \setminus u$ by adding the comparable vertex $u$ and the build-tree of $G$ is updated accordingly. Finally, the algorithm checks if $G$ has any cut vertices $z$ such that $G \setminus z$ has a component $Q$ that is a complete graph and such that $z$ is adjacent to all vertices in $Q$. If so, the algorithm recursively checks if $G \setminus Q$ is an OAT graph. If so, then $G$ is an OAT graph and can be constructed by attaching $Q$ to $z$ in $G \setminus Q$ and the build-tree is updated accordingly. The correctness of the algorithm follows from Lemma \ref{lem:op1}, Lemma \ref{lem:op2},  Lemma \ref{lem:op3}, and Lemma \ref{lem:op4}. 

Next we give implementation details and analyze the run-time of the algorithm. Let $G$ be the input graph with $n$ vertices and $m$ edges. In each recursive step of the algorithm, there are four conditions that may need to be checked.

\begin{case6}
First determine the connected components of $G$. This can be done using depth-first-search in $O(n + m)$ time.
\end{case6}

\begin{case6}
Next determine the connected components of the complement of $G$. This can be done in $O(n^2)$ time.
\end{case6}

\begin{case6}
Next, search for the biconnected components and the cut vertices of $G$. This can be done in $O(n+m)$ time using depth-first-search \cite{tarjan1972}. For each of the biconnected components $Q$ of $G$, check if the vertices form a clique. This can be done in $O(n_Q^2)$ time by checking if $Q$ has $\binom{n_Q}{2}$ edges. Therefore, this entire step can be done in $O(n^2)$ time.
\end{case6}

\begin{case6}
Finally search for a vertex $u$ comparable to another vertex $v$ of $G$. A brute force approach for finding such a pair would take $O(n^3)$ time at each recursive step. Instead, we use the adjacency matrix of $G$ and maintain it recursively to find a pair of comparable vertices at each step. Let $A(G)$ (or simply $A$) be the adjacency matrix of $G$ where the rows and columns of $A$ are indexed by the vertices of $G$. For $u, v \in V(G)$, the entry of $A$ at row $u$ and column $v$, denoted by $A[u,v]$, is 1 if $u$ and $v$ are adjacent and 0 otherwise. It is well known that $A^r[x,y]$ gives the number of paths of length at most $r$ from $x$ to $y$ in $G$. Then it is easy to see that for non-adjacent and distinct vertices $x,y \in V(G)$, $N(x) \subseteq N(y)$ if and only if $A^2[x,y] = d(x)$ (the degree of $x$) \cite{spinrad2003}. We note that $A^2$ can be computed in $O(n^\omega)$ time where the current best known value of $\omega$ is about 2.376 \cite{coppersmith1990}. However, to achieve the $O(n^3)$ running time of our algorithm, we only require that $A^2$ be computed in $O(n^3)$ time since, as we next show, $A^2$ can be updated in $O(n^2)$ time in each recursive step of the algorithm. Then at each step of the algorithm, we can find a pair of comparable vertices (if they exist) in $O(n^2)$ time by scanning the entries of $A^2$.

\begin{lemma}
Let $G$ be an OAT graph and let $A^2(G)$ be given. Then for each subgraph $H$ of $G$ considered in the recursive steps of the algorithm, $A^2(H)$ can be computed in $O(n^2)$ time.
\end{lemma}

\begin{proof}
We show how to update $A^2(G)$ in the recursive steps for each of the four operations.
\begin{case5}
Suppose $G$ is the disjoint union of the graphs $L$ and $R$. Then $A^2(L)$ is simply the submatrix of $A^2(G)$ with rows and columns corresponding to the vertices of $L$, and similarly for $A^2(R)$. 
\end{case5}

\begin{case5}
Suppose $G$ was constructed from the join of the graphs $L$ and $R$. Then $A^2(L)$ is the matrix obtained from $A^2(G)$ by deleting the rows and columns corresponding to the vertices of $R$ and subtracting $n_R$ from every entry. Similarly, $A^2(R)$ is the matrix obtained from $A^2(G)$ by deleting the rows and columns corresponding to the vertices of $L$ and subtracting $n_L$ from every entry.
\end{case5}

\begin{case5}
Suppose $G$ was constructed by adding a vertex $u$ comparable to a vertex of the OAT graph $H$. Then $A^2(H)$ can be obtained from $A^2(G)$ by deleting row $u$ and column $u$ and for each pair of vertices $x,y \in N(u)$ (not necessarily distinct), subtracting 1 from $A^2[x,y]$.
\end{case5}

\begin{case5}
Suppose $G$ was constructed from attaching the complete graph $Q$ to the vertex $v$ of the OAT graph $H$. Then $A^2(H)$ can be obtained from $A^2(G)$ by deleting the rows and columns corresponding to the vertices of $Q$ and be subtracting $n_Q$ from the entry $A^2[v,v]$. 
\end{case5}
Clearly, for each of these cases, the desired matrix can be obtained from $A^2(G)$ in $O(n^2)$ time.
\end{proof}
\end{case6}

The algorithm may be required to check these conditions $O(n)$ times and therefore the total run-time of the algorithm is $O(n^3)$. We note that changing the order in which the algorithm checks for each of the four conditions may improve the run-time.

\section{Conclusion}
In this paper, we introduced a class of graphs called OAT graphs defined by four simple operations. This class of graphs includes chordal bipartite graphs, compact graphs, $P_4$-sparse graphs, and some graphs which are not perfect but have equal chromatic number and clique number. We showed that for any $k$-colourable OAT graph $G$, the reconfiguration graph $\mathcal{R}_{k+1}(G)$ is connected with diameter $O(n^2)$. The proof of this can be converted into an algorithm that exhibits a recolouring sequence between any two $(k+1)$-colourings of $G$ in polynomial time. We also give an algorithm to recognize OAT graphs in $O(n^3)$ time and if the input graph $G$ is an OAT graph, the algorithm will return a build-tree of $G$. We end by stating two open problems.

\begin{problem}
Is it possible to include the operations of adding a true twin and adding a simplicial vertex to the operations defining OAT graphs and still have $\mathcal{R}_{k+1}(G)$ connected and of diameter $O(n^2)$?
\end{problem}
The class of graphs built from these six operations would then include distance-hereditary graphs and chordal graphs. Are there other simple operations that would allow for the inclusion of these graph classes?

\begin{problem}
Given a $k$-colourable weakly chordal graph $G$, what is the complexity of determining whether $\mathcal{R}_{k+1}(G)$ is connected?
\end{problem}

\end{document}